\documentclass[smallextended]{svjour3}       
\smartqed  
\newtheorem{fact}{Fact}

 \journalname{Acta Mathematica Vietnamica}
\begin{document}

\title{Codes induced by alternative codes 
}


\author{Ngo Thi Hien$^*$         \and
        Do Long Van$^\dag$ 
}

\authorrunning{Ngo Thi Hien and Do Long Van} 

\institute{$^*$Hanoi University of Science and Technology \at
              No. 1 Dai Co Viet Road, Hanoi, Vietnam \\
              \email{hien.ngothi@hust.edu.vn}           
           \and
           $^\dag$Institute of Mathematics, Vietnam Academy of Science and Technology \at
              No. 18 Hoang Quoc Viet Road, Hanoi, Vietnam \\
              \email{dlvan@math.ac.vn}           
}

\date{Received: date / Accepted: date}

\maketitle

\begin{abstract}
Alternative codes, an extension of the notion of ordinary codes, have been first introduced and considered by P. T. Huy et al. in 2004. As seen below, every alternative code, in its turn, defines an ordinary code. Such codes are called codes induced by alternative codes or alt-induced codes, for short. In this paper we consider these alt-induced codes and subclasses of them. In particular, characteristic properties of such codes are established, and an algorithm to check whether a finite code is alt-induced or not is proposed.
\end{abstract}

\section{Introduction}
The theory of length-variable codes has been initiated by M. P. Sch\"utzenberger in the 1950s and then developed by many others. This theory has now become a part of theoretical computer science and of formal languages, in particular. A code is a language such that every text encoded by words of the language can be decoded in a unique way or, in other words, every coded message admits only one factorization into code-words. Codes are useful in many areas of application such as information processing, data compression, cryptography, information transmission and so on. For background of the theory of codes we refer to \cite{BP85,JK97,Shyr91}.

As mentioned above, the definition of codes is essentially based on unambiguity of the (catenation) product of words. 
Different modifications of such a product may lead to different extensions of the notion of codes. Such an approach has been proposed by P.~T.~Huy et al. which deals with the so-called {\it even alternative codes} and their subclasses \cite{HN2004,VNH2010}. 
Throughout this paper, for simplicity, even alternative codes (see \cite{HN2004})  
are called simply {\it alternative codes} instead. 
As seen later, an alternative code is nothing but a pair $(X, Y)$ of languages such that $XY$ is a code and the product $XY$ is unambiguous (Lemma~\ref{L:XYeCodeCode}).
We say that the code $XY$ is induced by the alternative code $(X,Y)$. A code is said to be an {\it alt-induced code} if there exists an alternative code which induces it.

Obviously, every alt-induced code is a code. There exist, however, codes not being alt-induced ones. Characterizing alt-induced codes among codes is, therefore, a meaningful and interesting question.
Given a code $Z$, to prove that $Z$ is an alt-induced code, again by Lemma~\ref{L:XYeCodeCode}, we must first show that $Z$ can be factorized into a product of two languages $X$ and $Y$, $Z=XY$, and then show that the product $XY$ is unambiguous. The first phase of this seemingly relates to a more general question of factorising a language, a code in particular, into a product of simpler languages (see \cite{BP85,BDM2012,HSS2007,HV2006,JC2014,MPSchu66,W2010} and the references therein).

The paper is organized as follows. Section 2 presents basic notions, notations and facts which will be useful in the sequel.
In Section 3 the notion of alt-induced code is introduced. Several basic properties of these codes are shown (Proposition~\ref{P:XYicode-XpYs}, Proposition~\ref{L:Zn-icode}, Proposition~\ref{P:icode-Clocat}). Characterizations for prefix (suffix, bifix) alt-induced codes are established (Theorem~\ref{T:iCpsb-Char}).
A special subclass of alt-induced codes, namely that of {\it strong alt-induced codes} is considered in Section 4. Characterizations and properties of these codes are established (Proposition~\ref{P:XY-siCode}, Theorem~\ref{T:siaCode-Char}, Theorem~\ref{T:siCpsb-Char}, Theorem~\ref{T:rsiCode-Char}).
Finally, Section 5 is reserved to finite alt-induced codes. The main problem is to answer the question whether a given finite code is alt-induced (Theorem~\ref{T:notiCode}). The section ends with an algorithm (Algorithm FIC), with a exponential time complexity (Theorem~\ref{T:oFIC}), for testing whether a given finite code is alt-induced or not.

\section{Preliminaries}

Let $A$ throughout be a finite alphabet, i.e. a non-empty finite set of symbols, which are called letters. Let $A^*$ be the set of all finite words over $A$. The empty word is denoted by $\varepsilon$ and $A^{+}$ stands for $A^{*} \setminus \{\varepsilon\}$. The number of all the occurrences of letters in a word $u$ is the {\it length} of $u$, denoted by $|u|$. Any subset of $A^*$ is a {\it language} over $A$.
A language $X$ is a {\it code} if for any $n, m \geq 1$ and any $x_1,  \dots,  x_n, y_1,  \dots,  y_m \in X,$ the condition
$$x_1x_2  \dots  x_n = y_1y_2  \dots  y_m$$
implies $n = m$ and $x_i = y_i$ for $i=1, \dots, n$. Since $\varepsilon.\varepsilon = \varepsilon$, a code never contains the empty word $\varepsilon$.

	A word $u$ is called an {\it infix} (a {\it prefix}, a {\it suffix}) of a word $v$ if there exist words $x, y$ such that $v=xuy$ (resp., $v = uy$, $v= xu$). The infix (prefix, suffix) is {\it proper} if $xy \ne \varepsilon$ (resp., $y \ne \varepsilon$, $x \ne \varepsilon$).
The set of proper prefixes of a word $w$ is denoted by ${\rm Pref}(w)$. We denote by ${\rm Pref}(X)$ the set of all proper prefixes of the words in $X \subseteq A^*$. The notations ${\rm Suff} (w)$ and ${\rm Suff} (X)$ are defined in a similar way.
	
	For $X, Y \subseteq A^*$, the {\it product} of $X$ and $Y$ is the set $XY = \{xy \ | \ x \in X, y \in Y\}$. The product is said to be {\it unambiguous} if, for each $z \in XY$, there exists exactly one pair $(x,y) \in X \times Y$ such that $z = xy$. We also use the notations
	$$X^0 = \{\varepsilon\}, \ \ X^{n+1} = X^nX \ (n \geq 0).$$

	For $w \in A^*$, we define
$$w^{-1}X = \{u \in A^* \ | \ wu \in X\}, \ Xw^{-1} = \{u \in A^* \ | \ uw \in X\}.$$
These notations are extended to sets in a natural way:
$$X^{-1}Y = \bigcup_{x \in X}x^{-1}Y, \ XY^{-1} = \bigcup_{y \in Y}Xy^{-1}.$$

Let $(X,Y)$ be a pair of non-empty subsets of $A^+$, and let $u_1, u_2, \dots u_n \in X \cup Y, n \geq 2$. We say that $u_1u_2 \dots u_n$ is an {\it alternative factorization on $(X, Y)$} if $u_i \in X$ implies $u_{i+1} \in Y$ and $u_i \in Y$ implies $u_{i+1} \in X$ for all $i = 1, 2, \dots, n-1$.
Two alternative factorizations $u_1u_2 \dots u_n$ and $v_1v_2 \dots v_m$ on $(X, Y)$ are said to be {\it similar } if they both begin and end with words in the same set $X$ or $Y$.

\begin{definition} Let $X$ and $Y$ be two non-empty subsets of $A^+$. The pair $(X, Y)$ is called an {\it alternative code} if no word in $A^+$ admits two different similar alternative factorizations on $(X, Y)$.
\end{definition}

For more details of alternative codes and their subclasses we refer to \cite{Hf2017,HN2004,VNH2010}.

\medskip
Now we formulate, in the form of lemmas, several facts which will be useful in the sequel.

\begin{lemma}[{\cite{BP85}}] \label{L:XnCode}
If $X \subseteq A^+$ is a code, then $X^n$ is a code for all integers $n \geq 1$.
\end{lemma}

\begin{lemma}[Sardinas-Patterson's criterion {\cite {SP1953}}, see also {\cite{BP85,HH2012}}] \label{L:SP1953}
Let $X$ be a subset of $A^+$, and let
	$$U_1 = X^{-1}X \setminus \{\varepsilon\}, \ \ U_{n+1} = X^{-1}U_n \cup U_n^{-1}X \ {\rm ~for~}  n \geq 1.$$
Then, X is a code if and only if none of the sets $U_n$ defined above contains the empty word $\varepsilon$.
\end{lemma}

Recall that, a language $X \subseteq A^+$ is a {\it prefix code} ({\it suffix code}) if no word in $X$ is a proper prefix (resp., proper suffix) of another word in it, and $X$ is a {\it bifix code} if it is both a prefix code and a suffix code.
A prefix code (suffix code, bifix code) $X$ is {\it maximal} over $A$ if it is not properly contained in another  prefix code (resp., suffix code, bifix code) over $A$.
Prefix codes, suffix codes and bifix codes play a fundamental role in the theory of codes (see \cite{BP85,HV2006,Shyr91}).

As a consequence of Proposition 4.1 in \cite{BP85} we have

\begin{lemma} \label{L:CpClosed}
	Let $X$ and $Y$ be non-empty subsets of $A^+$. Then
\begin{enumerate}
\renewcommand{\labelenumi} {(\roman{enumi})} 
\item If $X$ and $Y$ are prefix codes (maximal prefix codes), then $XY$ is a prefix code (resp., maximal prefix code);
\item If $XY$ is a prefix code (suffix code), then $Y$ is a prefix code (resp., $X$ is a suffix code);
\item If $X$ is a prefix code and $XY$ is a maximal prefix code, then $X$ and $Y$ are both maximal prefix codes.
\end{enumerate}
\end{lemma}

A subset $X$ of $A^*$ is {\it thin } if there exists at least one word $w \in A^*$ which is not an infix of any word in $X$, i.e. $X \cap A^*wA^* = \emptyset$. Evidently, for any $X,Y \subseteq A^*$, if $XY$ is thin then $X$ and $Y$ are thin also. Conversely we have

\begin{lemma}[{\cite{BP85}}, page 65] \label{L:ThinProd}
	For any $X,Y \subseteq A^*$, if $X$ and $Y$ are both thin then their product $XY$ is thin too.
\end{lemma}
 
	Concerning the maximality of thin codes we have

\begin{lemma}[{\cite{BP85}}, Proposition 2.1, page 145] \label{L:ThinMb}
	Let $X$ be a thin subset of $A^+$. Then, $X$ is a maximal bifix code if and only if $X$ is both a maximal prefix code and a maximal suffix code.
\end{lemma}

A simple characterization of the unambiguity of a product of two languages is given in the following lemma.
\begin{lemma}[\cite{HN2004}, see also \cite{VNH2010}] \label{L:XYUnam}
	Let $X$ and $Y$ be non-empty subsets of $A^+$. Then, the product $XY$ is unambiguous if and only if $X^{-1}X \cap YY^{-1} \setminus \{\varepsilon\} = \emptyset$.
\end{lemma}

The following result claims a basic characterization for alternative codes.

\begin{lemma}[\cite{HN2004}, see also {\cite{VNH2010}}] \label{L:XYeCodeCode}
	Let $X$ and $Y$ be non-empty subsets of $A^+$. Then, $(X, Y)$ is an alternative code if and only if $XY$ is a code and the product $XY$ is unambiguous.
\end{lemma}

\section{Codes induced by alternative codes}

We introduce and consider in this section a new class of codes concerning alternative codes which are called alt-induced codes. Characterizations for prefix (suffix, bifix) alt-induced codes are established.

\begin{definition} \label{D:icode}
	A subset $Z$ of $A^+$ is called a {\it  code induced by an alternative code} ({\it alt-induced code}, for short) if there is an alternative code $(X, Y)$ over $A$ such that $Z = XY$.
\end{definition}

As usual, a language $Z$ is a {\it prefix (suffix, bifix) alt-induced code} if it is an alt-induced code as well as a prefix (resp., suffix, bifix) code.

\medskip
Let us take an example.

\begin{example}	\label{E:icode}
	Consider the sets $X = \{ab, abba\}$ and $Y =\{b\}$ over $A = \{a, b\}$. By virtue of Lemma~\ref{L:XYUnam} and Lemma~\ref{L:XYeCodeCode}, it is not difficult to check that $(X, Y)$ is an alternative code over $A$. Hence $Z = XY = \{abb, abbab\}$ is an alt-induced code and therefore it is a suffix alt-induced code.
\end{example}

Our purpose, in the rest of the paper, is to answer the question when a given code $Z$ is an alt-induced code. Let us begin with some simple cases.

\medskip
\begin{remark} \label{R:notSCode}
	Evidently, a code is not an alt-induced code if it contains at least one word with the length one. 
\end{remark}

\begin{proposition} \label{P:XYicode-XpYs}
	Let $Z$ be a code over $A$ such that $Z = XY$ with $\emptyset \neq X, Y \subseteq A^+$. If $X$ is a prefix code or $Y$ is a suffix code, then $Z$ is an alt-induced code.
\end{proposition}

\begin{proof}
	If $X$ is a prefix code, then $X^{-1}X =\{\varepsilon\}$, and therefore $X^{-1}X \cap YY^{-1} \setminus \{\varepsilon\} = \emptyset$. By Lemma~\ref{L:XYUnam}, the product $XY$ is unambiguous. Thus, by Lemma~\ref{L:XYeCodeCode}, $(X, Y)$ is an alternative code, and hence $Z = XY$ is an alt-induced code. Similarly for the case when $Y$ is a suffix code. \qed
\end{proof}

\begin{corollary} \label{C:notSCode}
	Let $Z \subseteq A^+$ be a code such that all the words of $Z$ have length greater than or equal to 2. Then, if all the words of $Z$ begin (end) with the same letter in $A$, then $Z$ is an alt-induced code.
\end{corollary}  

\begin{proof}
	Suppose that all the words of $Z$ begin with the same letter $a$ in $A$. Then, we have $Z = \{a\}Y$ with $\emptyset \neq Y \subseteq A^+$. Since $\{a\}$ is a prefix code, by Proposition~\ref{P:XYicode-XpYs}, $Z$ is an alt-induced code. The argument is similar for the other case of the corollary. \qed 
\end{proof}

\begin{remark} \label {R:SingAlph} If $A$ is a one-letter alphabet, $|A|=1$, then, as well-known, every subset $Z$ of $A^+$ is a code if and only if it is a singleton, $Z= \{w\}$. Therefore, by Remark~\ref{R:notSCode} and Corollary~\ref{C:notSCode}, $Z$ is an alt-induced code if and only if $w$ has the length greater than or equal to 2, $|w| \geq 2$.
\end{remark}

\begin{corollary} 
	Let $Z$ be a code over $A$ such that $Z = XY$ with $\emptyset \neq X, Y \subseteq A^+$. If $|X| = 2$ or $|Y| = 2$, then $Z$ is an alt-induced code. 
\end{corollary}

\begin{proof}
	We treat the case $|X| = 2$. If $X$ is a prefix code, then by Proposition~\ref{P:XYicode-XpYs}, $Z$ is an alt-induced code. Otherwise, since $|X| = 2$, we have $X = \{x, xu\}$ for some $x, u \in A^+$. Therefore, $Z = XY = \{x, xu\}Y = \{x\}\{Y, uY\}$, that is all the words of $Z$ begin with the same letter in $A$. Thus, by Corollary~\ref{C:notSCode}, $Z$ is also an alt-induced code. For the case $|Y| = 2$, the argument is similar.  \qed 
\end{proof}

\begin{proposition} \label{L:Zn-icode}
	If $Z \subseteq A^+$ is an alt-induced code, then $Z^n$ is also an alt-induced code for all integers $n \geq 1$.
\end{proposition}

\begin{proof}
	Since $Z$ is a code, by Lemma~\ref{L:XnCode}, $Z^n$ is a code for all integers $n \geq 1$. Now we prove by induction that $Z^n$ is also an alt-induced code. Indeed, for $n=1$ it is true by assumption. Suppose the assertion is already true for $n-1$ with $n \geq 2$. Put $X = Z^{n-1}, Y = Z$. Then, on one hand, as mentioned above, $XY = Z^n$ is a code. On the other hand, the unambiguity of the product $XY$ follows directly from the fact that $Z$ is a code. Thus, by Lemma~\ref{L:XYeCodeCode}, $(X,Y)$ is an alternative code and therefore $Z^n = XY$ is an alt-induce code.  \qed 
\end{proof}

	The following example shows that the product of two alt-induced codes is not, in general, an alt-induced code.

\begin{example}
	Let us consider the sets $Z = \{aa, baa\} = \{a, ba\}\{a\}$ and $Z' = \{aa, aab\} = \{a\}\{a, ab\}$ which are, as easily seen, alt-induced codes over $A = \{a, b\}$. Put 
	$$R = ZZ' = \{a^4, a^4b, ba^4, ba^4b\}.$$
The word $w = a^4ba^4b$, for example, has two distinct factorizations: $w = (a^4)(ba^4b) = (a^4b)(a^4b)$ on $R$.
Thus $R$ is not a code, and therefore not an alt-induced code either.
\end{example}

	For prefix (suffix, bifix) alt-induced codes, we have however

\begin{proposition} \label{P:icode-Clocat}
	The product of two prefix (suffix, bifix) alt-induced codes is a prefix (resp., suffix, bifix) alt-induced code.
\end{proposition}

\begin{proof}
	We treat only the case of prefix alt-induced codes. For the other cases the arguments are similar. Let $Z$ and $Z'$ be two prefix alt-induced codes. Then, by Lemma~\ref{L:CpClosed}(i), $ZZ'$ is a prefix code. Since $Z$ is a prefix code, by Proposition~\ref{P:XYicode-XpYs}, $ZZ'$ is a prefix alt-induced code.   \qed 
\end{proof}

Next, we exhibit characterizations for prefix (suffix, bifix) alt-induced codes. For this, we need two more auxiliary propositions.

\begin{proposition} \label{P:XYCpsb}
	Let $X$ and $Y$ be non-empty subsets of $A^+$.
\begin{enumerate}
\renewcommand{\labelenumi} {(\roman{enumi})} 
\item If $X$ and $Y$ are prefix (maximal prefix) codes, then $XY$ and $YX$ are prefix (resp., maximal prefix) codes, and the products $XY$ and $YX$ are unambiguous;
\item If $X$ and $Y$ are suffix (maximal suffix) codes, then $XY$ and $YX$ are suffix (resp., maximal suffix) codes, and the products $XY$ and $YX$ are unambiguous;
\item If $X$ and $Y$ are bifix (maximal bifix thin) codes, then $XY$ and $YX$ are bifix (resp., maximal bifix thin) codes, and the products $XY$ and $YX$ are unambiguous.
\end{enumerate}
\end{proposition}

\begin{proof}
	(i) If $X$ and $Y$ are prefix (maximal prefix) codes, then, by Lemma~\ref{L:CpClosed}(i), $XY$ and $YX$ are prefix (resp., maximal prefix) codes. Again because $X$ is a prefix code, it follows that $X^{-1}X =\{\varepsilon\}$, and therefore $X^{-1}X \cap YY^{-1} \setminus \{\varepsilon\} = \emptyset$. By Lemma~\ref{L:XYUnam}, the product $XY$ is unambiguous. Similarly, since $Y$ is a prefix code, it follows that $YX$ is unambiguous.

	(ii) The proof is similar as above in the item (i).

	(iii) If $X$ and $Y$ are bifix codes, then, by (i) and (ii) of the proposition, $XY$ and $YX$ are bifix codes and the products $XY$ and $YX$ are unambiguous. Now assume moreover that $X$ and $Y$ are maximal bifix codes and thin. Then, by Lemma~\ref{L:ThinMb}, $X$ and $Y$ are both maximal prefix codes as well as maximal suffix codes. Again by (i) and (ii) of the proposition, $XY$ and $YX$ are both maximal prefix codes as well as maximal suffix codes. By Lemma~\ref{L:ThinProd}, $XY$ and $YX$ are thin. Again by Lemma~\ref{L:ThinMb}, they are both maximal bifix codes.\qed 
\end{proof}

As the converse of Proposition~\ref{P:XYCpsb} we have

\begin{proposition} \label{P:XYCpsb-Conv}
	Let $X$ and $Y$ be non-empty subsets of $A^+$.
\begin{enumerate}
\renewcommand{\labelenumi} {(\roman{enumi})} 
\item If $XY$ and $YX$ are prefix (maximal prefix) codes, then $X$ and $Y$ are prefix (resp., maximal prefix) codes;
\item If $XY$ and $YX$ are suffix (maximal suffix) codes, then $X$ and $Y$ are suffix (resp., maximal suffix) codes;
\item If $XY$ and $YX$ are bifix (maximal bifix thin) codes, then $X$ and $Y$ are bifix (resp., maximal bifix thin) codes.
\end{enumerate}
\end{proposition}

\begin{proof}
	(i) If $XY$ and $YX$ are prefix codes, then, by Lemma~\ref{L:CpClosed}(ii), $Y$ and $X$ are prefix codes. If $XY$ and $YX$ are maximal prefix codes, then, again by Lemma~\ref{L:CpClosed}(ii) and then by Lemma~\ref{L:CpClosed}(iii), $X$ and $Y$ are maximal prefix codes.

	(ii) The proof is similar as for the item (i).

	(iii) Firstly, by (i) and (ii) of the proposition, $X$ and $Y$ are bifix codes if so are $YX$ and $XY$. Now, suppose that $XY$ and $YX$ are maximal bifix codes and thin. Then, by Lemma~\ref{L:ThinMb}, $XY$ and $YX$ are both maximal prefix and maximal suffix codes. Again by (i) and (ii) of the proposition, $X$ and $Y$ are both maximal prefix and maximal suffix codes. The thinness of $XY$ implies evidently the thinness of $X$ and $Y$. Hence, again by Lemma~\ref{L:ThinMb}, $X$ and $Y$ are maximal bifix codes.
\qed 
\end{proof}

As an immediate consequence of Proposition~\ref{P:XYCpsb} and Proposition~\ref{P:XYCpsb-Conv} we can now formulate the following result which resumes characterizations for prefix (suffix, bifix) alt-induced codes.

\begin{theorem} \label{T:iCpsb-Char}
	Let $X$ and $Y$ be non-empty subsets of $A^+$.
\begin{enumerate}
\renewcommand{\labelenumi} {(\roman{enumi})} 
\item $XY$ and $YX$ are prefix (maximal prefix) alt-induced codes if and only if $X$ and $Y$ are prefix (resp., maximal prefix) codes;
\item $XY$ and $YX$ are suffix (maximal suffix) alt-induced codes if and only if $X$ and $Y$ are suffix (resp., maximal suffix) codes;
\item $XY$ and $YX$ are bifix (maximal bifix thin) alt-induced codes if and only if $X$ and $Y$ are bifix (resp., maximal bifix thin) codes.
\end{enumerate}
\end{theorem}

	Note that Theorem~\ref{T:iCpsb-Char} provides us with a tool to construct (maximal) alt-induced codes. For example, just by taking product of two prefix (maximal prefix) codes, we obtain a prefix (maximal prefix) alt-induced code. As such, in some sense, the class of alt-induced codes is quite large.

\begin{example}
	We consider the sets $X = \{a^nba \ | \ n \geq 1\}, Y = \{b^ma \ | \ m \geq 1\}$ over $A = \{a, b\}$. Clearly, $X$ and $Y$ are prefix codes. Hence, by Theorem~\ref{T:iCpsb-Char}(i), $XY = \{a^nbab^ma \ | \ n, m \geq 1\}$ and $YX = \{b^ma^nba \ | \ n \geq 2, m \geq 1\}$ are prefix alt-induced codes.
\end{example}

\section{Strong alt-induced codes}

In this section we consider a special subclass of alt-induced codes which is introduced in the following definition. 

\begin{definition} \label{D:sicode}
	Let $X$ and $Y$ be non-empty subsets of $A^+$.
\begin{enumerate}
\renewcommand{\labelenumi} {(\roman{enumi})} 
\item An alternative code $(X, Y)$ is called a {\it strong alternative code} if it satisfies the following conditions
		$$X^{-1}(XY) \subseteq Y \ \ (1), \ \ (XY)Y^{-1} \subseteq X \ \ (2).$$
\item An alt-induced code $Z$ over $A$ is called a {\it strong alt-induced code} if there exists a strong alternative code $(X,Y)$ generating it, $Z=XY$.
\end{enumerate}
\end{definition}

\begin{remark}
	It is easy to check that the conditions (1) and (2) in Definition~\ref{D:sicode} are equivalent to the conditions 
		$$X^{-1}(XY) = Y \ \ (1'), \ \ (XY)Y^{-1} = X \ \ (2')$$
and these, in its turn, are equivalent to
		$$\forall x \in X: x^{-1}(XY) = Y \ \ (1''), \ \ \forall y \in Y: (XY)y^{-1} = X \ \ (2'').$$
\end{remark}

\begin{example}	\label{E:sicode}
	Consider the sets $X = \{a^nb \ | \ n \geq 1\}$ and $Y = \{b, bab\}$ over $A = \{a, b\}$. Then, we have $XY = \{a^nbb, a^nbbab \ | \ n \geq 1\}$. By Lemma~\ref{L:SP1953}, $XY$ is a code because $U_1 = \{ab\}, U_2 = \{b, bab\}, U_3 = \emptyset$.
Since $X$ is a prefix code, by Proposition~\ref{P:XYCpsb}(i), the product $XY$ is unambiguous. Hence, by Lemma~\ref{L:XYeCodeCode}, $(X, Y)$ is an alternative code, and therefore $XY$ is an alt-induced code.
On the other hand, it is easy to see that, for all $n\geq 1, a^nbba \in (XY)Y^{-1} \setminus X$, which means that the inclusion (2) in Definition~\ref{D:sicode} is not satisfied. Thus, $(X, Y)$ is not a strong alternative code, and therefore $XY$ is not a strong alt-induced code.
\end{example}

The following result is basic in characterizing the strong alternative codes and, therefore, the strong alt-induced codes.

\begin{proposition} \label{P:XY-siCode}
	Let $X$ and $Y$ be non-empty subsets of $A^+$.
\begin{enumerate}
\renewcommand{\labelenumi} {(\roman{enumi})} 
\item If $X$ is a prefix code, $Y$ is a suffix code and $XY$ is a code, then $(X, Y)$ is a strong alternative code;
\item If $(X, Y)$ is a strong alternative code, then $X$ is a prefix code and $Y$ is a suffix code.
\end{enumerate}
\end{proposition}

\begin{proof}
	(i) Suppose $X$ is a prefix code, $Y$ is a suffix code and $XY$ is a code. First, because $X$ is a prefix code, the product $XY$ is unambiguous. Thus, by Lemma~\ref{L:XYeCodeCode}, $(X, Y)$ is an alternative code. Now, we will prove that $(X, Y)$ is, moreover, a strong alternative code, that is we have to show that the inclusions (1) and (2) in Definition~\ref{D:sicode} must be satisfied. Indeed, assume the contrary that the inclusion (1), for example, is not true.  Then there must exist some word $u$ such that $u \in X^{-1}(XY)$ but $u \notin Y$. From $u \in X^{-1}(XY)$, it follows that there exist $x, x' \in X$ and $y \in Y$ such that $xu = x'y$. This, again because $X$ is prefix, implies $u = y$ and therefore $u \in Y$, a contradiction. Thus, the inclusion (1) must be true.
Similarly, the assumption that the inclusion (2) is not satisfied will also lead to a contradiction. As a consequence it must be true and therefore $(X, Y)$ is really a strong alternative code.

	(ii) Suppose $(X, Y)$ is a strong alternative code. Now assume that $X$ is not a prefix code. Then, there must exist $x, x' \in X$ such that $x = x'u$ with $u \neq \varepsilon$. Choosing any $y \in Y$ we have $xy = x'uy$ with $u \neq \varepsilon$. This implies  $uy \in X^{-1}(XY)$. If $uy \in Y$ then, from the unambiguity of the product $XY$, it follows $uy=y$, a contradiction. Thus, $uy \notin Y$ which means that the inclusion (1) in Definition~\ref{D:sicode} is not satisfied. This contradicts the assumption that $(X,Y)$ is a strong alternative code. Thus, $X$ must be a prefix code. The fact that $Y$ is a suffix code can be proved in a similar way, where the inclusion (2) is used instead of (1). \qed
\end{proof}

As an immediate consequence of Definition~\ref{D:sicode} and Proposition~\ref{P:XY-siCode}, we obtain the following characterization for strong alt-induced codes as well as strong alternative codes.

\begin{theorem} \label{T:siaCode-Char}
	Let $X$ and $Y$ be non-empty subsets of $A^+$. The following conditions are equivalent:
\begin{enumerate}
\renewcommand{\labelenumi} {(\roman{enumi})} 
\item $XY$ is a strong alt-induced code;
\item $(X, Y)$ is a strong alternative code;
\item $X$ is a prefix code, $Y$ is a suffix code and $XY$ is a code.
\end{enumerate}
\end{theorem}

\begin{example}
	Consider the sets $X = \{a^nb \ | \ n \geq 1\}, Y = \{b, bba\}$ over $A = \{a, b\}$. Evidently, $X$ is a prefix code not being a suffix code while $Y$ is a suffix code not being a prefix code. Then, by virtue of Lemma~\ref{L:SP1953}, the product $XY = \{a^nbb, a^nbbba \ | \ n \geq 1\}$ is easily verified to be a code: $U_1 = \{ba\}, U_2 = \emptyset$. According to Theorem~\ref{T:siaCode-Char}, $(X,Y)$ is a strong alternative code and $XY$ is a strong alt-induced code.
\end{example}

The following example shows that in Theorem~\ref{T:siaCode-Char} the requirement \lq\lq$XY$ is a code\rq\rq\ in the item (iii) is essential.

\begin{example}
	Consider the sets $X = \{a^nb \ | \ n \geq 1\} , Y = \{ba^m \ | \ m \geq 1\}$ over $A = \{a, b\}$. Clearly, $X$ is a prefix code not being a suffix code while $Y$ is a suffix code not being a prefix code. But their product $Z = XY = \{a^nbba^m \ | \ n, m \geq 1\}$ is not a  code because the word $w = abbaaabba$, for example, has two distinct factorizations on $Z$:
$$w = (abba)(aabba) = (abbaa)(abba).$$
\end{example}

As a consequence of Theorem~\ref{T:siaCode-Char}, Lemma~\ref{L:CpClosed} and Lemma~\ref{L:ThinMb} we have the following characterization for prefix (suffix, bifix) strong alt-induced codes.

\begin{theorem} \label{T:siCpsb-Char}
	Let $X$ and $Y$ be non-empty subsets of $A^+$.
\begin{enumerate}
\renewcommand{\labelenumi} {(\roman{enumi})} 
	\item $XY$ is a prefix (maximal prefix) strong alt-induced code if and only if $X$ is a prefix (resp.,  maximal prefix) code and $Y$ is a bifix code (resp., $Y$ is both a maximal prefix code and a bifix code);
	\item $XY$ is a suffix (maximal suffix) strong alt-induced code if and only if $X$ is a bifix code (resp., $X$ is both a maximal suffix code and a bifix code) and $Y$ is a suffix (resp., maximal suffix) code;
	\item $XY$ is a bifix (maximal bifix thin) strong alt-induced code if and only if $X$ and $Y$ are bifix (resp., maximal bifix thin) codes.
\end{enumerate}
\end{theorem}

\begin{proof}
	(i) Suppose $XY$ is a prefix strong alt-induced code. Then, by Theorem~\ref{T:siaCode-Char}, $X$ is a prefix code, $Y$ is a suffix code and $XY$ is a prefix code. Therefore, by Lemma~\ref{L:CpClosed}(ii), $Y$
must be a prefix code and therefore a bifix code.
Now, suppose that $XY$ is moreover a maximal prefix strong alt-induced code. Then, $XY$ is, in particular, a maximal prefix code and, by the above, $X$ is a prefix code and $Y$ is a bifix code. Therefore, by Lemma~\ref{L:CpClosed}(iii), both $Y$ and $X$ are maximal prefix codes.

	Conversely, suppose now $X$ is a prefix code and $Y$ is a bifix code. Then, by Lemma~\ref{L:CpClosed}(i), $XY$ is a prefix code. Therefore, by Theorem~\ref{T:siaCode-Char}, $XY$ is a prefix strong alt-induced code. Next, suppose moreover that $X$ is a maximal prefix code and $Y$ is both a maximal prefix code and a bifix code. Then, on one hand, by the above, $XY$ is a strong alt-induced code and, on the other hand, again by Lemma~\ref{L:CpClosed}(i), $XY$ is a maximal prefix code. Hence, $XY$ is a maximal prefix strong alt-induced code.

	(ii) The proof is similar as in the item (i).

	(iii) It follows immediately from Lemma~\ref{L:ThinMb}, and the items (i) and (ii) of the theorem.
\qed 
\end{proof}


\begin{example}
	Consider the sets $X = \{a^nb \ | \ n \geq 1\}, Y = \{ba^mb \ | \ m \geq 1\}$ over $A = \{a, b\}$. Evidently, $X$ is a prefix code and $Y$ is a bifix code. Hence, by Theorem~\ref{T:siCpsb-Char}(i), $XY = \{a^nbba^mb \ | \ n, m \geq 1\}$ is a prefix strong alt-induced code.
\end{example}

In the framework of regular languages, the strong alt-induced codes have the following interesting property.

\begin{theorem} \label{T:rsiCode-Char}
	If $Z \subseteq A^+$ is a regular strong alt-induced code, then there exists only a finite number of strong alternative codes inducing $Z$.
\end{theorem}

\begin{proof} 
	Firstly, recall that, for any language $Z$ over $A$, the syntactic congruence of $Z$, denote by $\cong_Z$, is defined as follows
	$${\rm For~}  u, v\in A^*: u \cong_Z v \ {\rm if~ and~ only~ if~}  \forall x, y\in A^*: xuy \in Z \Leftrightarrow xvy \in Z.$$

Next, it is easy to verify that, for any word $x \in A^*$, $x^{-1}Z$ is a union of equivalence classes of $\cong_Z$. It follows therefore that, for any $X,Y \subseteq A^*$, $X^{-1}Z$ and $ZY^{-1}$ are also  unions of equivalence classes of $\cong_Z$.

Now, suppose $Z$ is a regular strong alt-induced code. On one hand, because $Z$ is regular, $\cong_Z$ has finite index, i.e. there exists only a finite number of equivalence classes according to the congruence $\cong_Z$. On the other hand, since $Z$ is a strong alt-induced code, there must exist  a strong alternative code $(X,Y)$ such that $Z = XY$. By Definition~\ref{D:sicode}, are verified the following equalities:
		$$X^{-1}(XY) = Y \ \ (1'), \ \ (XY)Y^{-1} = X \ \ (2').$$
Thus, by the above, $X$ and $Y$ are unions of equivalence classes of the congruence $\cong_Z$. Because $\cong_Z$ has finite index, it follows that the number of such pairs $(X, Y)$ must be finite.
 \qed 
\end{proof}

\section{Finite alt-induced codes}

This section is reserved to consider alt-induced codes in the framework of finite codes. Some properties, allowing to answer the question, whether a given finite code is an alt-induced code or not, are exhibited. An algorithm to test whether a given finite code is an alt-induced code is proposed. 
 
We denote the cardinality of a set $X$ by $|X|$. The following fact is evident.

\begin{fact} \label{F:XYunam-ZXY} Let $X$ and $Y$ be non-empty finite subsets of $A^+$. Then, the product $XY$ is unambiguous if and only if $|XY| = |X|.|Y|$.
\end{fact}

From Definition~\ref{D:icode} and Fact~\ref{F:XYunam-ZXY} it follows directly

\begin{corollary} \label{C:eCode-ZXY}
	If $Z$ is a finite alt-induced code generated by an alternative code $(X,Y)$, $Z=XY$, then $|Z| = |X|.|Y|$.
\end{corollary}

Evidently also
\begin{fact} \label{F:ZXY-ctXY}
	If $X$ and $Y$ are non-empty subsets of $A^+$ and $Z = XY$, then $X \subseteq \bigcap_{y \in Y}{Zy^{-1}}$ and $Y \subseteq \bigcap_{x \in X}{x^{-1}Z}$.
\end{fact}

\begin{remark} \label{R:ZXY-ctXY}
The converse inclusions of those in Fact~\ref{F:ZXY-ctXY} are not true in general even with assumption that the sets $X, Y$ and $Z$ are finite codes. Consider, for example, the suffix codes $X = \{ab, ab^3, b^2a\}$, $X' = \{a, ab^2\}$ and the prefix codes $Y = \{a^2, ab, ba^2, bab\}$, $Y' = \{a, ba\}$ over $A = \{a, b\}$. Then, we have
\begin{quotation}
\noindent
	$Z = XY = \{aba^2, abab$, $ab^2a^2$, $ab^2ab$, $ab^3a^2$, $ab^3ab$, $ab^4a^2$, $ab^4ab$,

\hskip 1.46cm $b^2a^3$, $b^2a^2b$, $b^2aba^2$, $b^2abab\}$,

\noindent
$Z' = X'Y' = \{a^2, aba, ab^2a, ab^3a\}$.
\end{quotation}

It is easy to verify that $Z$ is a prefix code, $Z'$ is a bifix code, and $\bigcap_{y \in Y}{Zy^{-1}} = \{ab, ab^2, ab^3, b^2a\} \not\subseteq X$, $\bigcap_{y \in Y'}{Z'y^{-1}} = \{a, ab, abb\} \not\subseteq X'$.
\end{remark}

Let $A = \{a_1, a_2, \dots, a_k\}, k \geq 2$, and $Z \subseteq A^+$. For every $i$, we denote by $Z_{a_i}$ the set of all the words in $Z$ beginning with the letter $a_i$, namely
		$$Z_{a_i} = \{w \in Z \ | \ w = a_iu, u\in A^*\}$$
where $i = 1,2, \dots, k$.
Evidently, the non-empty $Z_{a_i}$ constitute a partition of $Z$. Namely 
$$Z = \bigcup_{i=1}^k{Z_{a_i}}, \ Z_{a_i} \cap Z_{a_j} = \emptyset \ (i \neq j) \ \ (3).$$

The following result, whose proof is based on Fact~\ref{F:ZXY-ctXY}, will be useful in the sequel.

\begin{proposition} \label{L:Ysub-uZ}
	Let $A = \{a_1, a_2, \dots, a_k\}, k \geq 2$, and $Z \subseteq A^+, Z = \bigcup_{i=1}^k{Z_{a_i}}$, $Z_{a_i} \neq \emptyset$, $i = 1, 2, \dots, k$.
If $Z = XY$ with $\emptyset \neq X,Y \subseteq A^+$, then $\forall w \in Z_{a_i},\exists u \in {\rm Pref} (w)$ such that $Y \subseteq u^{-1}Z_{a_i}, i = 1, 2, \dots, k$. Moreover, $Y = u^{-1}Z_{a_i}$ if $X$ is a prefix code.
\end{proposition}

\begin{proof}
	Suppose $Z \subseteq A^+, Z = \bigcup_{i=1}^k{Z_{a_i}}$, $Z_{a_i} \neq \emptyset$, $i = 1, 2, \dots, k$, and $Z = XY$ with $\emptyset \neq X,Y \subseteq A^+$. Then, from $XY = \bigcup_{i=1}^k{Z_{a_i}}$, it follows that $Z_{a_i} = X_{a_i}Y$ for all $i = 1, 2, \dots, k$, where $X = \bigcup_{i=1}^k{X_{a_i}}$.
On the other hand, since $Z = XY$, by Fact~\ref{F:ZXY-ctXY}, $Y \subseteq \bigcap_{x \in X}{x^{-1}Z}$. Therefore $Y \subseteq x^{-1}Z, \forall x \in X$, and hence $Y \subseteq x^{-1}Z, \forall x \in X_{a_i}, i = 1, 2, \dots, k$. Because $x^{-1}Z = x^{-1}Z_{a_i}, \forall x \in X_{a_i}$, it follows that $Y \subseteq x^{-1}Z_{a_i}, \forall x \in X_{a_i}$, $i = 1, 2, \dots, k$.
Thus $\forall w \in Z_{a_i}$, take $u \in {\rm Pref} (w)$ such that $u \in X_{a_i}$, we have $Y \subseteq u^{-1}Z_{a_i}, i = 1, 2, \dots, k$. 

	If moreover $X$ is a prefix code, then $x^{-1}Z_{a_i} = x^{-1}(X_{a_i}Y) = (x^{-1}X_{a_i})Y = \{\varepsilon\}Y = Y$, $\forall x \in X_{a_i}, i = 1, 2, \dots, k$.
Hence, $Y = u^{-1}Z_{a_i}$ where $u \in {\rm Pref} (w)$, $w \in Z_{a_i},  i = 1, 2, \dots, k$.  \qed 
\end{proof}

From now on, we suppose always that the alphabet $A$ consists of at least two letters, $|A| \geq 2$. Then, a finite code $Z$ over $A$ is called to be of {\it standard form} if all the words of $Z$ have the length greater than or equal to 2, and it is not the case that all the words of $Z$ begin (end) with the same letter in $A$.

The following two results will be useful in checking whether a finite code of standard form is an alt-induced code or not.
As usual, $\gcd(n_1, n_2, \dots, n_k)$ denotes the greatest common divisor of $n_1, n_2, \dots, n_k$.

\begin{theorem} \label{T:notiCode}
	Let $A = \{a_1, a_2, \dots, a_k\}, k \geq 2$, and $Z \subseteq A^+$ is a finite code of standard form, $Z = \bigcup_{i=1}^k{Z_{a_i}}$. 
If $\gcd(|Z_{a_1}|, |Z_{a_2}|, \dots, |Z_{a_k}|) = 1$, then $Z$ is not an alt-induced code.
\end{theorem}

\begin{proof}
	Suppose the contrary that $Z$ is an alt-induced code. Then, there is an alternative code $(X, Y)$ such that $Z = XY$. The equality $XY = \bigcup_{i=1}^k{Z_{a_i}}$ implies $Z_{a_i} = X_{a_i}Y$ for all $i = 1, 2, \dots, k$, where $X = \bigcup_{i=1}^k{X_{a_i}}$. On the other hand, by (3), we have $|Z| = \sum_{i=1}^k{|Z_{a_i}|}$. Since $(X, Y)$ is an alternative code and $Z$ is a code of standard form,  by Corollary~\ref{C:eCode-ZXY}, $|Z| = |X|.|Y|$ with $|X|, |Y| \geq 2$. All this implies that $|Z_{a_i}| = |X_{a_i}|.|Y|$ for all $i = 1, 2, \dots, k$, which contradicts the assumption $\gcd(|Z_{a_1}|, |Z_{a_2}|, \dots, |Z_{a_k}|)  = 1$. Thus, $Z$ is not an alt-induced code. \qed
\end{proof}

\begin{corollary} \label{C:notiCode}
	If $Z$ is a finite code of standard form over $A$ and $|Z|$ is prime, then $Z$ is not an alt-induced code.
\end{corollary}

\begin{proof}
	Assume $A = \{a_1, a_2, \dots, a_k\}, k \geq 2$, and $Z \subseteq A^+$ is a finite code of standard form. Then, by (3), we have $|Z| = \Sigma_{i=1}^k|Z_{a_i}|$. Since $|Z|$ is prime, it follows that $\gcd(|Z_{a_1}|, |Z_{a_2}|, \dots, |Z_{a_k}|)  = 1$. By Theorem~\ref{T:notiCode}, $Z$ is not an alt-induced code. \qed
\end{proof}

The following example shows that the class of codes mentioned in Theorem~\ref{T:notiCode} is strictly larger than that in Corollary~\ref{C:notiCode}.

\begin{example}
	Consider the set $Z = \{abc, acb, bac, bca, bbac, cab, cba, caab\}$ over $A = \{a, b, c\}$. It is easy to see that $Z$ is a prefix code of standard form, $Z = Z_a \cup Z_b \cup Z_c$, where
	$$Z_a = \{abc, acb\}, Z_b = \{bac, bca, bbac\}, Z_c = \{cab, cba, caab\}.$$
We have $|Z| = 8$ which is composite, but $\gcd(|Z_a|, |Z_b|, |Z_c|) = \gcd(2,3,3) = 1$. Hence, by Theorem~\ref{T:notiCode}, $Z$ is not an alt-induced code. 
\end{example}

From Corollary~\ref{C:eCode-ZXY}, Fact~\ref{F:ZXY-ctXY}, Proposition~\ref{L:Ysub-uZ} and Theorem~\ref{T:notiCode}, we can exhibit the following algorithm for testing whether a given finite code is alt-induced or not.
As usual, when $Z$ is a finite set of $A^+$, $\min Z$ denotes the minimal wordlength of $Z$.
The set of all subsets of $Z$ is denoted by $2^Z$.

\medskip
{\noindent\bf Algorithm FIC} {\it (A test for finite alt-induced codes)}

{\it Input:} A finite code $Z$ of standard form and $Z = \bigcup_{i=1}^k{Z_{a_i}}.$

{\it Output:} $Z$ is an alt-induced code or not.
 

	1. If $\gcd(|Z_{a_1}|,|Z_{a_2}|, \dots, |Z_{a_k}|) = 1$ then go to Step 4.

	2. Choose $w \in Z_{a_t}, 1 \leq t \leq k$, such that $|w| = \min Z$;

\leftskip0.45cm 
		Set $P_w = {\rm Pref}(w)\setminus \{\varepsilon\}$

		and $D = \{d \geq 2 \ | \ d {\rm ~is~a~common~divisor~of~} |Z_{a_1}|, |Z_{a_2}|, \dots, |Z_{a_k}|\}$.

\leftskip 0cm 
	3. While $P_w \neq \emptyset$ do \{
	

\leftskip0.8cm 
	   Take $u \in P_w$; Set $S = u^{-1}Z_{a_t}$ and $Q = \{U \in 2^S \ | \ |U| \in D\}$;

	   While $Q \neq \emptyset$ do \{

\leftskip1.2cm 
		Take $Y \in Q$;

		Set $P = \bigcap_{y \in Y} {Zy^{-1}}$ and $R = \{V \in 2^P \ | \ |V| = |Z|/|Y|\}$;

		While $R \neq \emptyset$ do \{

\leftskip1.6cm 
			Take $X \in R$;

			If $Z =XY$ then go to Step 5 else $R := R \setminus X$;

\leftskip1.2cm 
		\} 

		$Q := Q \setminus Y$;

\leftskip0.8cm 
	   \} 

	   $P_w := P_w \setminus \{u\}$;

\leftskip0.45cm 
	\} 

\leftskip 0cm 
	4. $Z$ is not an alt-induced code; STOP.

	5. $Z$ is an alt-induced code; STOP.

\medskip

Correctness of FIC's algorithm follows immediately from Theorem~\ref{T:notiCode} and the following result.

\begin{proposition} \label{P:iCode}
	Let $Z$ be a finite code of standard form over $A=\{a_1, \dots,a_k\}$, $k \geq 2$, such that $\gcd(|Z_{a_1}|,|Z_{a_2}|,\dots,|Z_{a_k}|) > 1$,  let $w \in Z_{a_t}, 1 \leq t \leq k$, with $|w| = \min Z$, and $P_w$, $Q$ and $R$ are defined as above. Then, $Z$ is an alt-induced code if and only if there exist $u \in P_w$, $Y \in Q$ and $X \in R$ such that $Z = XY.$
\end{proposition}

\begin{proof}
	Suppose $Z$ is an alt-induced code. Then, there is an alternative code $(X, Y)$ such that $Z = XY$, with $\emptyset \neq X,Y \subseteq A^+$. Therefore, by Proposition~\ref{L:Ysub-uZ}, there exists $u \in P_w$ such that $Y \subseteq u^{-1}Z_{a_t}$ with $1 \leq t \leq k$ and $|w| = \min Z$.
 On the other hand, by (3), we have $|Z| = \sum_{i=1}^k{|Z_{a_i}|}$. Since $(X, Y)$ is an alternative code and $Z$ is a finite code of standard form,  by Corollary~\ref{C:eCode-ZXY}, $|Z| = |X|.|Y|$ with $|X|, |Y| \geq 2$. All this implies that $|Y|$ is a common divisor of $|Z_{a_1}|, |Z_{a_2}|, \dots, |Z_{a_k}|$. Thus, $Y \in Q$.
Next, since $Z = XY$ with $\emptyset \neq X,Y \subseteq A^+$, by Fact~\ref{F:ZXY-ctXY}, $X \subseteq \bigcap_{y \in Y} {Zy^{-1}}$. Hence, we have $X \in R$ because $|Z| = |X|.|Y|$.

	Conversely, suppose there exist $u \in P_w$, $Y \in Q$ and $X \in R$ such that $Z = XY.$ Then, 
$|Z| = |XY| = |X|.|Y|$ with $\emptyset \neq X,Y \subseteq A^+$. Therefore, by Fact~\ref{F:XYunam-ZXY}, the product $XY$ is unambiguous. Thus, by Lemma~\ref{L:XYeCodeCode}, $(X, Y)$ is an alternative code. Hence, $Z$ is an  alt-induced code. \qed
\end{proof}

Let us take some examples.

\begin{example}
	Consider the set $Z = \{a^3, a^2b, ba^2, b^9\}$ over $A = \{a, b\}$. Clearly, $Z$ is a code of standard form, $Z = Z_a \cup Z_b$ with $Z_a = \{a^3, a^2b\}$ and $Z_b = \{ba^2, b^9\}$. By Algorithm~FIC, we have:

	1. Since $|Z| = 4$ and $\gcd(|Z_a|, |Z_b|) = 2$, we go to Step 2.

	2. Choose $w = a^3 \in Z_a$ with $|w| = 3 = \min Z$. Set $P_w = \{a, aa\}$ 

\leftskip0.45cm 
		and $D = \{d \geq 2 \ | \ d \ {\rm ~is~a~common~divisor~of~} |Z_a|,|Z_b|\} = \{2\}$.

\leftskip 0cm 
	3. While $P_w \neq \emptyset$ do \{
	

\leftskip0.5cm 
	   3.1. Take $a \in P_w$, we have $S = a^{-1}Z_a = \{aa, ab\}$ and $Q = \{\{aa, ab\}\}$; 
		
\leftskip0.8cm 
	   While $Q \neq \emptyset$ do \{

\leftskip1.2cm 
		Take $Y = \{aa, ab\} \in Q$;

		Set $P = Z(aa)^{-1} \cap Z(ab)^{-1} = \{a, b\}\cap \{a\} = \{a\}$
		and $R = \{\{a\}\}$;

		While $R \neq \emptyset$ do \{

\leftskip1.6cm 
			Take $X = \{a\} \in R$;

			Since $Z \neq XY$, which implies $R := R \setminus X = \emptyset$;

\leftskip1.2cm 
		\} 

		$Q := Q \setminus Y = \emptyset$;

\leftskip0.8cm 
	   \} 

	   $P_w := P_w \setminus \{a\} = \{aa\} \neq \emptyset$;

\leftskip0.5cm 
	   3.2. Take $aa \in P_w$, we have $S = (aa)^{-1}Z_a = \{a, b\}$ and $Q = \{\{a, b\}\}$; 
		
\leftskip0.8cm 
		While $Q \neq \emptyset$ do \{

\leftskip1.2cm 
		Take $Y = \{a, b\} \in Q$;

		Set $P = Za^{-1} \cap Zb^{-1} = \{a^2, ba\}\cap \{a^2, b^8\} = \{a^2\}$
		and $R = \{\{a^2\}\}$;

		While $R \neq \emptyset$ do \{

\leftskip1.6cm 
			Take $X = \{a^2\} \in R$;

			Since $Z \neq XY$, which implies $R := R \setminus X = \emptyset$;

\leftskip1.2cm 
		\} 

		$Q := Q \setminus Y = \emptyset$;

\leftskip0.8cm 
	   \} 

	   $P_w := P_w \setminus \{aa\} = \emptyset$; We go to Step~4.

\leftskip0.45cm 
	\} 

\leftskip 0cm 
	4. $Z$ is not an alt-induced code, and the algorithm ends.
\end{example}

\begin{example}
	Again consider the set $Z$ in Remark~\ref{R:ZXY-ctXY}. Evidently, $Z$ is a code of standard form, $Z = Z_a \cup Z_b$, where

	$Z_a = \{aba^2, abab, ab^2a^2, ab^2ab, ab^3a^2, ab^3ab, ab^4a^2, ab^4ab\}$

	$Z_b = \{b^2a^3, b^2a^2b, b^2aba^2, b^2abab\}$.

\noindent
By Algorithm~FIC, we have:

	1. Since $|Z| = 12$ and $\gcd(|Z_a|, |Z_b|) = 4$, we go to Step 2.

	2. Choose $w = aba^2 \in Z_a$ with $|w| = 4 = \min Z$. Set $P_w = \{a, ab, aba\}$ 

\leftskip0.45cm 
		and $D = \{d \geq 2 \ | \ d \ {\rm ~is~a~common~divisor~of~} |Z_a|,|Z_b|\} = \{2, 4\}$.

\leftskip 0cm 
	3. While $P_w \neq \emptyset$ do \{
	

\leftskip0.8cm 
	   Take $ab \in P_w$, we have $S = \{a^2, ab, ba^2, bab, b^2a^2, b^2ab, b^3a^2, b^3ab\}$

	   and $Q = \{U \in 2^S \ | \ |U| \in \{2, 4\}\}$; 
		
	   While $Q \neq \emptyset$ do \{

\leftskip 1.2cm 
		Take $Y = \{a^2, ab, ba^2, bab\} \in Q$;

		Set $P = \bigcap_{y \in Y} {Zy^{-1}} = \{ab, ab^2, ab^3, b^2a\}$
		
		and $R = \{V \in 2^P \ | \ |V| = 3\}$;

		While $R \neq \emptyset$ do \{

\leftskip 1.6cm 
			Take $X = \{ab, ab^3, b^2a\} \in R$;

			Because $Z = XY$, we go to Step 5;

\leftskip 1.2cm 
		\} 

\leftskip 0.8cm 
	   \} 

\leftskip 0.45cm 
	\} 

\leftskip 0cm 
	5. $Z = XY = \{ab, ab^3, b^2a\}\{a^2, ab, ba^2, bab\}$ is an alt-induced code, and the algorithm ends.
\end{example}

\begin{theorem} \label{T:oFIC}
	Given a finite code $Z$ of standard form over $A = \{a_1, a_2, \dots, a_k\}$, $k \geq 2$, we can determine whether $Z$ is an alt-induced code or not in $O(km^2n4^n)$ worst-case time, where $m = \min Z$ and $n = \min \{|Z_{a_1}|,|Z_{a_2}|,\dots , |Z_{a_k}|\}$.
\end{theorem}

\begin{proof}
	The best-case for the algorithm is when $\gcd(|Z_{a_1}|, |Z_{a_2}|, \dots, |Z_{a_k}|) = 1$, which takes $O(hk)$ steps to perform the task, where $h$ is the number of digits in the number $n$ (see~\cite{Mollin2008}, pages 21--22).
   
	In Step 2, we can choose $w \in Z_{a_t}$ such that $|w| = m = \min Z$, $|Z_{a_t}| = \min \{|Z_{a_1}|,|Z_{a_2}|,\dots , |Z_{a_k}|\}$, and the set $P_w$ has at most $m-1$ words. Then, in Step 3, for each word $u$ in $P_w$, it takes $O(2^{2n})$ worst-case time to perform for $Q$ and $R$, and $O(|Y|.|Z|) = O(m.kn)$ worst-case time in finding $P$. Thus, the total running time for determining, whether $Z$ is an alt-induced code or not, is $O(m).O(2^{2n}).O(m.kn) = O(km^2n4^n)$ in the worst-case. \qed
\end{proof}

\begin{acknowledgements}
	The authors would like to thank the colleagues in Seminar Mathematical Foundation of Computer Science at Institute of Mathematics, Vietnam  Academy of Science and Technology for attention to the work. Especially, the authors express our sincere thanks to Dr. Nguyen Huong Lam and Dr. Kieu Van Hung for their useful discussions.
\end{acknowledgements}


\end{document}